\documentclass[12 pt]{extarticle}
\usepackage{amsmath}
\usepackage{amssymb}
\usepackage[a4paper, total={6in, 9in}]{geometry}
\usepackage{amsthm,url,tikz}
\usepackage[utf8]{inputenc}
\usepackage[english]{babel}
\usepackage{graphicx}

\usepackage{makecell}

\usepackage{cases} 
\usepackage{enumerate}
\usepackage{hyperref}
\usepackage{adjustbox}
\numberwithin{equation}{section}
\usepackage{makecell}
\usepackage{bm}
\usepackage{xcolor}
\newtheorem{theorem}{Theorem}[section]
\newtheorem{proposition}[theorem]{Proposition}
\newtheorem{corollary}[theorem]{Corollary}
\newtheorem{lemma}[theorem]{Lemma}
\theoremstyle{definition}
\newtheorem{definition}[theorem]{Definition}

\newtheorem{example}[theorem]{Example}
\newtheorem{remark}[theorem]{Remark}

\usepackage[nottoc,notlot,notlof]{tocbibind} 

\title{Some constructions of  non-generalized Reed-Solomon MDS Codes}
\author{Kanat Abdukhalikov \\	
Department of Mathematical Sciences, \\
UAE University, PO Box 15551, Al Ain, UAE\\
Email: abdukhalik@uaeu.ac.ae \bigskip  \\  
Cunsheng Ding \\
Department of Computer Science and Engineering, \\
The Hong Kong University of Science and Technology, Hong Kong, China\\
Email: cding@ust.hk \bigskip \\
Gyanendra K. Verma \\	
Department of Mathematical Sciences, \\
UAE University, PO Box 15551, Al Ain, UAE\\
Email:  gkvermaiitdmaths@gmail.com}

\date{}

\begin{document}
	\maketitle
\begin{abstract} 
We investigate two classes of extended codes and provide necessary and sufficient conditions for these codes to be non-GRS MDS codes. We also determine the parity check matrices for these codes. Using the connection of MDS codes with arcs in finite projective spaces, we give a new characterization of $o$-monomials.
\end{abstract}

\section{Introduction}
Let $\mathbb{F}_q$ be a finite field of order $q$, where $q$ is a prime power. A linear subspace of $\mathbb{F}_q^n$ is called a linear code of length $n$ over $\mathbb{F}_q$. 
A code $C$ is said to be an $[n,k,d]_q$ code if $C$ is linear code of length $n$ and dimension $k$ with the minimum distance $d$.
For an $[n,k,d]_q$ code, Singleton bound states that $d\leq n-k+1$. Codes that attain the Singleton bound are called Maximum Distance Separable (MDS) codes. These codes are optimal for error detection and correction, as they attain the largest possible minimum distance for a given length and dimension.
MDS codes also have applications in combinatorial designs and are closely related to finite geometry \cite{Macwilliams1977,Hirschfeld1998}. 
A $k$-dimensional linear MDS code over $\mathbb{F}_q$ is  equivalent to an arc in the projective space $PG(k-1,q)$. 

In general, constructing MDS codes with large lengths remains a challenging problem. This difficulty is closely tied to the well-known MDS conjecture, which imposes strict limitations on the maximum length of an MDS code relative to the size of the underlying finite field. 
The famous MDS conjecture states that for an $[n,k,d]$ MDS code where $d \ge 3$, the length $n$ satisfies $n \le q+1$, unless $k \in \{3, q-1\}$ and $q=2^h$ in which case $n \le q+2$. The MDS conjecture has been verified for $q$ prime \cite{Ball2012}. By this conjecture, if $4 \le k \le q-3$, then a $k$-dimensional linear MDS code of length $n$ should satisfy inequality $n \le q+1$. The MDS conjecture asserts that, within that range, one cannot do better than the generalized Reed–Solomon code. 
GRS codes are the only known class of MDS codes over $\mathbb{F}_q$, which exist for all lengths $n\le q+1$ and dimension $k \le n$. 
MDS conjecture is given by Segre \cite{Segre1955} in the context of finite projective geometry (for details, see \cite{Ball2019,Macwilliams1977}).  
Generalized Reed-Solomon (GRS) codes have found several practical applications due to their algebraic structures and properties. The dual of the GRS code is also a GRS code. The length of GRS codes is restricted by field size, that is, any GRS code over a finite field $\mathbb{F}_q$ can have a maximum length of $q+1$.  In \cite{Seroussi1986}, Seroussi and Roth showed that a GRS code of dimension  $k$ with $2\leq k\leq n-\lfloor \frac{q-1}{2}\rfloor$ ($k\neq 3$ if $q$ is even)  can be extended to an MDS code if and only if the resulting extended code is also a GRS. Moreover, GRS codes of length $q+1$ and $2\leq k\leq \lfloor \frac{q}{2}\rfloor+2$ ($k\neq 3$ if $q$ is even) do not admit further MDS extensions.

A non-GRS MDS code is an MDS code that is not equivalent to a GRS code. The construction of non-GRS MDS codes is inherently driven by both the fundamental goal of classifying MDS codes and their significant applications in cryptography (see \cite{beelen2018,Lavauzelle2019,Zemor2015}). In 1989, Roth and Lempel \cite{Roth1989}  constructed non-GRS MDS codes via generator matrices. In \cite{Beelen2017,Beelen2022}, Beelen et al. introduced a class of new codes called twisted Reed-Solomon (TRS) codes. These are not MDS codes in general. However, the class of TRS code exhibits several subclasses of non-GRS MDS codes. Jin et al. \cite{Jin2024} introduced a general method to construct MDS codes by putting some conditions on the set of evaluation polynomials and the set of evaluation points. Using the Schur square technique, they proved that these MDS codes are non-GRS MDS codes. In \cite{Liu2024}, Liu et al. provided some generic constructions of MDS codes via TRS and generalized TRS codes and obtained some families of non-GRS MDS codes. Recently, several families of non-GRS MDS codes have been constructed using various techniques in \cite{Chen2022,Li2024,Wu2024,Wu2021,Zhi2025}. 

In this work, we contribute further to this area by presenting constructions of non-GRS MDS codes, derived through the study of extended codes of codes introduced in \cite{Jin2024}. The contributions of this work are as follows. In Section \ref{knownGRS} 
we prove that the sufficient conditions obtained in \cite{Jin2024} for codes to be MDS are necessary as well. Furthermore, we determine the parity check matrices for non-GRS codes obtained in \cite{Jin2024}. 
In Sections \ref{firstextension} and \ref{secondextension} we construct two classes of non-GRS MDS codes by extending the code from \cite{Jin2024}. We also determine the parity check matrices for the newly constructed codes. In Section \ref{oplynomials} we explore the connection between 3-dimensional MDS codes and $o$-polynomials and provide a new characterization of $o$-monomials in terms of complete symmetric functions.

\section{Notation and Background}\label{notations}
In this section, we give some definitions, notations, and preliminaries that will be used throughout the paper. 

For a positive integer $k$, let
$$V_k=\{f(x)\in \mathbb{F}_q[x]\mid  \deg f(x)\leq k-1\}.$$
\begin{definition}
For $n\leq q+1$, let $\Lambda=\{\alpha_1,\alpha_2,\dots \alpha_n\}\subseteq \mathbb{F}_q\cup\{\infty\}$. The generalized Reed-Solomon (GRS) code is defined as 
$$C_{GRS}=\{(v_1f(\alpha_1),v_2f(\alpha_2),\dots,v_nf(\alpha_n)) \mid  f(x)\in V_k\}$$
for some $v=(v_1,v_2,\dots,v_n)\in \mathbb{F}_q^n$ with $v_i\neq 0$ for all $1\leq i\leq n$. 
Here the element $f(\infty)$ is defined as the coefficient of $x^{k-1}$ in the polynomial $f(x)$, $\deg f(x) \le k-1$. 
\end{definition}
In particular, when $v=(1,1,\dots,1)$, the code is called a Reed-Solomon (RS) code.

Note that several authors do not include the point $\infty$ in the definition of GRS code, they call it extended GRS code.  However, a more natural definition of  the GRS code, which is entirely equivalent to the classical one, is obtained by evaluating
 homogeneous polynomials $f(x,y) \in \mathbb{F}_q[x, y]$ of degree $k - 1$, at the points of the projective line, as follows (see \cite{Ball2023}). 
 Let $\alpha_n=\infty$. Then $\alpha_n$ corresponds to $(1:0)$ in the projective line, thus
 \begin{align*}
     \begin{split}
         C_{GRS}=&\{(v_1f(\alpha_1:1),v_2f(\alpha_2:1),\dots,v_nf(1:0) \mid\\
         &f(x,y)\in \mathbb{F}_q[x,y], f(x,y) \text{ is homogeneous}, \deg f(x,y)= k-1\},
     \end{split}
 \end{align*}
where if $f(x)=f_0+f_1x+\dots+f_{k-2}x^{k-2}+f_{k-1}x^{k-1}\in V_k$, then the corresponding homogeneous polynomial  $f(x,y)\in \mathbb{F}_q[x,y]$ is  
 $$f(x,y)=f_0y^{k-1}+f_1xy^{k-2}+\dots+f_{k-2}x^{k-2}y+f_{k-1}x^{k-1},$$
 and $f(\infty)=f(1:0)=f_{k-1}$. 

It is well established that GRS codes are MDS codes. If $\Lambda=\{\alpha_1,\alpha_2,\dots \alpha_n\}\subseteq \mathbb{F}_q\cup\{\infty\}$, then generator matrices for GRS codes are given by 
$$G_{GRS}=\begin{bmatrix}
    v_1&v_2& \dots &v_n\\
    v_1\alpha_1& v_2\alpha_2& \dots&v_n\alpha_n\\
    v_1\alpha_1^2& v_2\alpha_2^2& \dots&v_n\alpha_n^2\\
    \vdots &\vdots &\dots&\vdots\\
    
    v_1\alpha_1^{k-1}& v_2\alpha_2^{k-1}& \dots&v_n\alpha_n^{k-1}
\end{bmatrix}_{k\times n},$$
where in case $\alpha_i =\infty$ for some $1\leq i\leq n$, the $i$-th column of the matrix $G_{GRS}$ is the coloumn $(0,0,\dots,v_i)^T$. 

\begin{definition}\label{nongrs}
A non-GRS MDS code is an MDS code that is not equivalent to a GRS code. 
\end{definition}
The following lemma is a well known characterization of MDS codes in terms of generator matrices. 
\begin{lemma}\cite[p.319]{Macwilliams1977}\label{mdsnonsingular}
 Let $C$ be an $[n,k]_q$ code with a generator matrix $G$. Then $C$ is MDS if and only if all $k\times k$ sub-matrices of $G$ are non-singular.   
\end{lemma}

Let $\Lambda=\{\alpha_1,\alpha_2,\dots,\alpha_n\}\subseteq \mathbb{F}_q$. Let $M_{\Lambda}$ be the square Vandermonde matrix of order $n$ defined by 
$$M_{\Lambda}=\begin{bmatrix}
    1&1&\dots&1\\
    \alpha_1&\alpha_2&\dots&\alpha_n\\
    \vdots&\vdots&\vdots&\vdots\\
    \alpha_1^{n-2}&\alpha_2^{n-2}&\dots&\alpha_n^{n-2}\\
    \alpha_1^{n-1}&\alpha_2^{n-1}&\dots&\alpha_n^{n-1}
\end{bmatrix}.$$
It is well known that the determinant of $M_\Lambda$ is  $\det(M_{\Lambda})=\prod_{1\leq i< j\leq n}(\alpha_j-\alpha_i)$.
Now, we fix some notation that will be used throughout the paper. For an integer $t\geq 0$, the $t^{\rm{th}}$ elementary symmetric function in $n$ variables is defined as
$$\sigma_t(x_1,x_2,\dots x_n)=\sum_{1\leq j_1<j_2<\dots<j_t\leq n}x_{j_1}x_{j_2}\cdots x_{j_t}$$
with $\sigma_0(x_1,x_2,\dots,x_n)=1$ and $\sigma_t(x_1,x_2,\dots x_n)=0$ for all $t>n$.

Let $S_t(x_1,x_2,\dots,x_n)$ be the complete symmetric polynomial of degree $t$ defined by 
$$S_t(x_1,x_2,\dots,x_n)=\sum_{t_1+t_2+\dots+t_n=t,\ t_i\geq 0}x_1^{t_1}x_2^{t_2}\dots x_n^{t_n}.$$ 
Let $\Lambda=\{\alpha_1,\alpha_2\dots,\alpha_n\}\subseteq \mathbb{F}_q$. By Cramer's rule, the solution of the system of linear equations 
$$\begin{bmatrix}
    1&1& \dots &1\\
    \alpha_1& \alpha_2& \dots&\alpha_n\\
    \alpha_1^2& \alpha_2^2& \dots&\alpha_n^2\\
    \vdots &\vdots &\dots&\vdots\\
    
    \alpha_1^{n-2}& \alpha_2^{n-2}& \dots&\alpha_n^{n-2}\\
    \alpha_1^{n-1}& \alpha_2^{n-1}& \dots&\alpha_n^{n-1}\\
\end{bmatrix}\begin{bmatrix}
    x_1\\
    x_2\\
    x_3\\
    \vdots\\
    x_{n-1}\\
      x_{n}\\
\end{bmatrix}=\begin{bmatrix}
    0\\
    0\\
    0\\
    \vdots\\
    0 \\
    1
\end{bmatrix}$$
is given by 
$$u_i=\prod_{1\leq j\leq n, j\neq i}(\alpha_i-\alpha_j)^{-1}, \ 1\leq i\leq n.$$ 
 
 We denote $S_t(\alpha_1,\alpha_2,\dots,\alpha_n)$ by $S_t$ with $S_0=1$ and $\sigma_t(\alpha_1,\alpha_2,\dots,\alpha_n)$ by $\sigma_t$. For uniformity, we write $S_t=0$ for $t<0$. There is a fundamental relation between the elementary symmetric functions and the complete symmetric functions (for instance, see \cite[Page 21, Eq. $2.6'$]{Macdonald1995}
\begin{equation}\label{eqA}
    \sum_{t=0}^N(-1)^{t}\sigma_tS_{N-t}=0, \ \ \  \text{for all } N\geq 1.
\end{equation}
Also, for $1\leq i\leq n$ and $r\geq 1$, define 

\begin{equation}\label{lambdari}
    \Lambda_{r,i}=\sum_{j=0}^r(-1)^j\sigma_j\alpha_i^{(n-k)+(r-1)-j}.
    \end{equation}

The following proposition determines the determinant of a matrix obtained by removing an arbitrary row from a Vandermonde matrix and adding another row with an exponent equal to the order of the matrix. We frequently utilize this result to study MDS codes presented in this paper.

\begin{proposition}\cite[Sec. 337]{Muir1960}\label{detk-r}
For $1\leq r\leq n-1$, let 
\begin{equation*}
   M_n^{(r)}=\begin{bmatrix}
    1&1& \dots &1\\
    \alpha_1& \alpha_2& \dots&\alpha_n\\
    \alpha_1^2& \alpha_2^2& \dots&\alpha_n^2\\
    \vdots &\vdots &\dots&\vdots\\
    \alpha_1^{n-r-1}& \alpha_2^{n-r-1}& \dots&\alpha_n^{k-r-1}\\
    \alpha_1^{n-r+1}& \alpha_2^{n-r+1}& \dots&\alpha_n^{n-r+1}\\
    \vdots &\vdots &\dots&\vdots\\
    \alpha_1^{n-1}& \alpha_2^{n-1}& \dots&\alpha_n^{n-1}\\
    \alpha_1^{n}& \alpha_2^{n}& \dots&\alpha_n^{n}\\
\end{bmatrix}_{n\times n}.
\end{equation*}
Then \begin{equation*}
    \det(M_n^{(r)})= \det(M_{\Lambda})\cdot \sigma_r(\alpha_1,\alpha_2,\dots,\alpha_n).
    \end{equation*}
\end{proposition}

Let $n$ and $h$ be positive integers. Define a generalized Vandermonde matrix $G_h$ as 
$$G_h=\begin{bmatrix}
    1&1&\dots&1\\
    \alpha_1&\alpha_2&\dots&\alpha_n\\
    \vdots&\vdots&\vdots&\vdots\\
    \alpha_1^{n-2}&\alpha_2^{n-2}&\dots&\alpha_n^{n-2}\\
    \alpha_1^h&\alpha_2^h&\dots&\alpha_n^h
\end{bmatrix}.$$
The following lemma determines the determinant of the matrix $G_h$.
\begin{lemma}\cite[Proposition 3.1]{Marchi2001}\label{detgh}
With notations as above, for $h\geq n-1$, we have
  $$\det(G_h)=S_{h-n+1}(\alpha_1,\alpha_2,\dots,\alpha_n)\cdot\det(M_{\Lambda}).$$   
\end{lemma}

\begin{lemma}\label{paritylemma}
    Let $\Lambda=\{\alpha_1,\alpha_2,\dots,\alpha_n\}\subseteq \mathbb{F}_q$ and $u_i=\prod_{1\leq j\leq n, j\neq i}(\alpha_i-\alpha_j)^{-1}$ for $1\leq i\leq n$. Then 
$$\sum_{i=1}^nu_i\alpha_i^h=\begin{cases}
  0 &\text{ if } 1\leq h\leq n-2\\
  S_{h-n+1}(\alpha_1,\alpha_2,\dots,\alpha_n) &\text{ if }  h\geq n-1
\end{cases}.$$
\end{lemma} 
\begin{proof}
From \cite[Lemma 5]{Li2008}, we have     $$\sum_{i=1}^nu_i\alpha_i^h=\det(G_h)/\det(M_{\Lambda})=0 \text{ if } 1\leq h\leq n-2.$$
The rest of the proof follows from Lemma \ref{detgh}.
\end{proof}

\section{A class of non-GRS MDS code}\label{knownGRS}
In this section, we discuss a class of non-GRS MDS codes introduced in \cite{Jin2024}. The authors only investigated sufficient conditions for these codes to be MDS. We prove that the sufficient condition obtained in \cite{Jin2024} is necessary as well. Furthermore, we determine the parity check matrix for these codes. 

 Let $k$ and $r$ be positive integers with $1\leq r\leq k-1$. Define $$V_{r,k}=\left \{\sum_{i=0}^{k-r-1}f_ix^i+\sum_{i=k-r+1}^kf_ix^i\in \mathbb{F}_q[x]\ \mid f_i\in \mathbb{F}_q\right \}.$$
Let $n\leq q$ and $\Lambda=\{\alpha_1,\alpha_2,\dots \alpha_n\}\subseteq \mathbb{F}_q$. Define 
$$C_{r,k}=\{ (f(\alpha_1),f(\alpha_2),\dots,f(\alpha_n))\ \mid  f(x)\in V_{r,k}\}.$$ Then $C_{r,k}$ is a linear code over $\mathbb{F}_q$ with parameters $[n,k]$ and a generator matrix
$$G_{r,k}=\begin{bmatrix}
    1&1& \dots &1\\
    \alpha_1& \alpha_2& \dots&\alpha_n\\
    \alpha_1^2& \alpha_2^2& \dots&\alpha_n^2\\
    \vdots &\vdots &\dots&\vdots\\
    \alpha_1^{k-r-1}& \alpha_2^{k-r-1}& \dots&\alpha_n^{k-r-1}\\
    \alpha_1^{k-r+1}& \alpha_2^{k-r+1}& \dots&\alpha_n^{k-r+1}\\
    \vdots &\vdots &\dots&\vdots\\
    \alpha_1^{k-1}& \alpha_2^{k-1}& \dots&\alpha_n^{k-1}\\
    \alpha_1^{k}& \alpha_2^{k}& \dots&\alpha_n^{k}\\
\end{bmatrix}_{k\times n}.$$

\begin{proposition}\cite[Proposition VI.1]{Jin2024}\label{sufcrkmds}
  Let $\Lambda=\{\alpha_1,\alpha_2,\dots \alpha_n\}\subseteq \mathbb{F}_q$, $6\leq 2k\le n$  and $1\leq r\leq k-1$ such that
  \begin{equation*}
      \sum_{1\leq j_1<\dots <j_r\leq k}\alpha_{i_{j_1}}\alpha_{i_{j_2}}\dots\alpha_{i_{j_r}}\neq 0\tag{$\ast$} \end{equation*}
for all pairwise distinct elements $\alpha_{i_{j}}$ 
with $1 \leq i_1 < i_2 < \dots < i_k \leq n$. Then the code $C_{k,r}$ is an $[n,k]$ non-GRS MDS code.
\end{proposition}
The condition $(\ast)$ means that for all $ k$-subsets of $\{\alpha_1,\alpha_2,\dots,\alpha_n\}\subseteq \mathbb{F}_q$, the sum of all $r$-element-products in the set is nonzero.

\begin{theorem}\label{iffcrkmds}
     Let $\Lambda=\{\alpha_1,\alpha_2,\dots \alpha_n\}\subseteq \mathbb{F}_q$, $6\leq 2k\le n$ and $1\leq r\leq k-1$. Then the code $C_{k,r}$ generated by $G_{r,k}$ is an $[n,k]$ MDS code if and only if
  \begin{equation*}
      \sum_{1\leq j_1<\dots <j_r\leq k}\alpha_{i_{j_1}}\alpha_{i_{j_2}}\dots\alpha_{i_{j_r}}\neq 0\tag{$\ast$} \end{equation*}
for all pairwise distinct elements $\alpha_{i_{j}}$ 
with $\{ i_1, i_2,\dots,i_k\}\subseteq \{1,2,\dots,n\}$. Moreover, in this case the code is non-GRS.
\end{theorem}
\begin{proof}
The condition $(\ast)$ and non-GRS follows from Proposition \ref{sufcrkmds}. Let $C_{r,k}$ be an MDS code. Let $\{i_1,i_2,\dots,i_k\}\subseteq \{1,2,\dots,n\}$ be an arbitrary set with cardinality $k$ and
$$B=\begin{bmatrix}
    1&1& \dots &1\\
    \alpha_{i_1}& \alpha_{i_2}& \dots&\alpha_{i_k}\\
    \alpha_{i_1}^2& \alpha_{i_2}^2& \dots&\alpha_{i_k}^2\\
    \vdots &\vdots &\dots&\vdots\\
    \alpha_{i_1}^{k-r-1}& \alpha_{i_2}^{k-r-1}& \dots&\alpha_{i_k}^{k-r-1}\\
    \alpha_{i_1}^{k-r+1}& \alpha_{i_2}^{k-r+1}& \dots&\alpha_{i_k}^{k-r+1}\\
    \vdots &\vdots &\dots&\vdots\\
    \alpha_{i_1}^{k-1}& \alpha_{i_2}^{k-1}& \dots&\alpha_{i_k}^{k-1}\\
    \alpha_{i_1}^{k}& \alpha_{i_2}^{k}& \dots&\alpha_{i_k}^{k}\\
\end{bmatrix}_{k\times k}$$ be the corresponding   $k\times k$ sub-matrix of $G_{r,k}$. Since any $k\times k$ sub-matrix of $G_{r,k}$ is non-singular, we have $\det (B)\neq 0$. By Proposition \ref{detk-r}, we have 
$$ \det(B)= \prod_{1\leq s<t\leq k}(\alpha_{i_t}-\alpha_{i_s})\cdot \sum_{1\leq j_1<j_2<\dots<j_r\leq k}\alpha_{i_{j_1}}\alpha_{i_{j_2}}\dots \alpha_{i_{j_r}}. $$
Also, $\alpha_i$'s are distinct. Therefore, $\det(B)\neq 0$ if and only if
$$ \sum_{1\leq j_1<j_2<\dots<j_r\leq k}\alpha_{i_{j_1}}\alpha_{i_{j_2}}\dots \alpha_{i_{j_r}}\neq 0.$$
This completes the proof.
\end{proof}

\begin{theorem}\label{paritygrk}
Let 
$$H_{r,k}=\begin{bmatrix}
    u_1&u_2&\dots&u_n\\
u_1\alpha_1&u_2\alpha_2&\dots&u_n\alpha_n\\
u_1\alpha_1^2&u_2\alpha_2^2&\dots&u_n\alpha_n^2\\
\vdots&\vdots&\vdots&\vdots\\
u_1\alpha_1^{n-k-2}&u_2\alpha_2^{n-k-2}&\dots&u_n\alpha_n^{n-k-2}\\
u_1\Lambda_{r,1}&u_2\Lambda_{r,2}&\dots&u_n\Lambda_{r,n}\\
\end{bmatrix}.$$
Then $H_{r,k}$ is a parity check matrix for the code $C_{r,k}$ generated by $G_{r,k}$.
\end{theorem}
\begin{proof}
Let $g_l$ and $h_m$ be the $l^{\rm{th}}$ and $m^{\rm{th}}$ rows of matrices $G_{r,k}$ and $H_{r,k}$, respectively. We show that $g_l h_m^T=0$ for $1\leq l\leq k$ and $1\leq m\leq n-k$. Using Lemma \ref{paritylemma}, we conclude that
if $1\leq l\leq k-r,\  1\leq m\leq n-k-1$, then
$$g_l h_m^T=
     \sum_{i=1}^n u_i\alpha_i^{l+m-2}=0.$$
     If $k-r+1\leq l\leq k, \ 1\leq m\leq n-k-1$, then
     $$g_lh_m^T=\sum_{i=1}^n u_i\alpha_i^{l+m-1}=0.$$
If $1\leq l\leq k-r$ and $m=n-k$, then
$$\sum_{i=1}^nu_i\alpha_i^{l-1}\Lambda_{r,i}=\sum_{i=1}^nu_i\alpha_i^{l-1}\sum_{j=0}^r(-1)^j\sigma_j\alpha_i^{(n-k)+(r-1)-j}=\sum_{j=0}^r(-1)^j\sigma_j\sum_{i=1}^nu_i\alpha_i^{(l-1)+(n-k)+(r-1)-j}=0,$$
since $(l-1)+(n-k)+(r-1)-j\leq \ n-2$ for $1\leq l\leq k-r$ and $0\leq j\leq r$.\\
If $k-r+1\leq l\leq k$ and $m=n-k$, then
\begin{align*}
    \begin{split}
\sum_{i=1}^nu_i\alpha_i^{l}\Lambda_{r,i}&=\sum_{i=1}^nu_i\alpha_i^{l}\sum_{j=0}^r(-1)^j\sigma_j\alpha_i^{(n-k)+(r-1)-j}\\
    &=\sum_{j=0}^r(-1)^j\sigma_j\sum_{i=1}^nu_i\alpha_i^{l+(n-k)+(r-1)-j} \\
    &=\sum_{j=0}^{l-k+r}(-1)^j\sigma_j\sum_{i=1}^nu_i\alpha_i^{l+(n-k)+(r-1)-j}+\sum_{j=l-k+r+1}^r(-1)^j\sigma_j\sum_{i=1}^nu_i\alpha_i^{l+(n-k)+(r-1)-j},\\
&=\sum_{j=0}^{l-k+r}(-1)^j\sigma_j\sum_{i=1}^nu_i\alpha_i^{l+(n-k)+(r-1)-j}+0, 
    \end{split}
    \end{align*}
since  $l+(n-k)+(r-1)-j\leq n-2$ for  $j\geq l-k+r+1$.
 Note that $l+(n-k)+(r-1)-j\geq n-1$. Therefore, again by Lemma \ref{paritylemma}, we have 
$$\sum_{i=1}^nu_i\alpha_i^{l}\Lambda_{r,i}=\sum_{j=0}^{l-k+r}(-1)^j\sigma_jS_{l-k+r-j}=0,$$ 
where the last equality holds by Eq \ref{eqA}, since $l-k+r\geq 1$.
Thus, every row of $H_{r,k}$ is orthogonal to each row of $G_{r,k}$.
Now, we show that the rows of $H_{r,k}$ are linearly independent. It is easy to see that the first $(n-k-1)$ rows of $H_{r,k}$ are linearly independent. Let $h_{n-k}$ be $(n-k)^{\rm{th}}$-row of $H_{r,k}$. Let $D=(\alpha_1^{k-r},\alpha_2^{k-r},\dots,\alpha_{n}^{k-r})\in \mathbb{F}_q^n$. Then $D$ is orthogonal to each row $h_1,h_2,\dots,h_{n-k-1}$ of $H_{r,k}$, that is, $Dh_m^T=0$ for $1\leq m\leq n-k-1$. But
\begin{align*}
    \begin{split}
        Dh_{n-k}^T=&\sum_{i=1}^nu_i\alpha_i^{k-r}\sum_{j=0}^r(-1)^j\sigma_j\alpha_i^{(n-k)+(r-1)-j}=\sum_{i=1}^nu_i\sum_{j=0}^r(-1)^j\sigma_j\alpha_i^{n-1-j}.
    \end{split}
\end{align*}
By Lemma \ref{paritylemma}, 
$$Dh_{n-k}^T=\sum_{j=0}^r(-1)^j\sigma_j\sum_{i=1}^nu_i\alpha_i^{n-1-j}=\sum_{i=1}^nu_i\sigma_0\alpha_i^{n-1}=1.$$
That is, $Dh_{n-k}^T\neq 0$. Thus, $h_{n-k}\notin  \text{span}\{h_1,h_2,\dots,h_{n-k-1}\}$. 
Therefore, the rows of $H_{r,k}$ are linearly independent. Hence, $H_{r,k}$ is a parity check matrix for the code $C_{r,k}$.
\end{proof}

\section{First class of extended non-GRS MDS codes}\label{firstextension}
In this section, we construct a class of codes generated by the matrix obtained by adding one column $(0,0,\dots,1)^T$ in the matrix $G_{r,k}$. First, we determine the parity check matrix for these codes. Then we show that these codes are exactly the extended codes of $C_{r,k}$ and establish necessary and sufficient conditions for these codes to be non-GRS MDS codes.  First, we define extended codes.
\begin{definition}[Extended code]
     Let $w=(w_1,w_2,\dots,w_n)$ be a non-zero element in $\mathbb{F}_q^n$ and $C$ be an $[n,k,d]_q$ code. Then 
$$\overline{C}(w)=\left \{ (c_1,c_2,\dots,c_n,c_{n+1}=\sum_{i=1}^n w_ic_i )|\ (c_1,c_2,\dots,c_n)\in C\right \}$$
is an $[n+1,k,\overline{d}]_q$ code with $\overline{d}=d$ or $\overline{d}=d+1$, called the extended code of $C$ with respect to $w$.
\end{definition}
 We say a code $C'$ is an extended code of $C$ if $C'=\overline{C}(w)$ for some non-zero $w\in \mathbb{F}_q^n$. If $G$ and $H$ are generator and parity-check matrices for $C$, then a generator and parity-check  matrix for $\overline{C}(w)$ are
$G_w=(G\ Gw^T)$ and $H_w=\begin{bmatrix}
    H& \bm{0}^T\\
    w&-1
\end{bmatrix}$.

The following lemma is an easy observation.
\begin{lemma}\cite[Lemma 5]{Wu2024}\label{extendnonrs}
   If the extended code $\overline{C}(w)$ is monomial equivalent to an GRS code, then the original code $C$ is also monomial equivalent to an GRS code. 
\end{lemma} 
 Let  
$$G_1=\begin{bmatrix}
    1&1& \dots &1 &0\\
    \alpha_1& \alpha_2& \dots&\alpha_n&0\\
    \alpha_1^2& \alpha_2^2& \dots&\alpha_n^2&0\\
    \vdots &\vdots &\dots&\vdots&\vdots\\
    \alpha_1^{k-r-1}& \alpha_2^{k-r-1}& \dots&\alpha_n^{k-r-1}&0\\
    \alpha_1^{k-r+1}& \alpha_2^{k-r+1}&\dots&\alpha_n^{k-r+1}&0\\
    \vdots &\vdots &\dots&\vdots&\vdots\\
    \alpha_1^{k-1}& \alpha_2^{k-1}& \dots&\alpha_n^{k-1}&0\\
    \alpha_1^{k}& \alpha_2^{k}& \dots&\alpha_n^{k}&1\\
\end{bmatrix}_{k\times (n+1)}$$ and $C_1$ be the linear code generated by $G_1$. Then the code $C_1$ can be seen as an evaluation code as follows:
\begin{equation*}
    \begin{split}
     C_1=&\{(f(\alpha_1),f(\alpha_2),\dots, f(\alpha_n), f_k)|\ f(x)\in V_{r,k} \}.
    \end{split}
\end{equation*}

\begin{theorem}\label{parityg1}
Let 
$$H_1=\begin{bmatrix}
    u_1&u_2&\dots&u_n&0\\
u_1\alpha_1&u_2\alpha_2&\dots&u_n\alpha_n&0\\
u_1\alpha_1^2&u_2\alpha_2^2&\dots&u_n\alpha_n^2&0\\
\vdots&\vdots&\vdots&\vdots&\vdots\\
u_1\alpha_1^{n-k-2}&u_2\alpha_2^{n-k-2}&\dots&u_n\alpha_n^{n-k-2}&0\\
u_1\Lambda_{r,1}&u_2\Lambda_{r,2}&\dots&u_n\Lambda_{r,n}&0\\
u_1\alpha_1^{n-k-1}& u_2\alpha_2^{n-k-1}&\dots&u_n\alpha_n^{n-k-1}&-1
\end{bmatrix}.$$
Then $H_1$ is a parity check matrix for the code $C_1$ generated by $G_1$. 
\end{theorem}
\begin{proof}
It is easy to see that the rows of $H_1$ are linearly independent, and the first $n-k$ rows of $H_1$ are orthogonal to each row of $G_1$ from Theorem \ref{paritygrk}. We show that the $(n-k+1)^{\rm{th}}$-row of $H_1$ is also orthogonal to each row of $G_1$. Note that for $m=n-k+1$, we have 
 \begin{align*}
  g_l h_m^T= \begin{cases}
      \sum_{i=1}^nu_i\alpha_i^{(l-1)+(n-k-1)}& \text{if } 1\leq l\leq k-r\\
      \sum_{i=1}^nu_i\alpha_i^{l+(n-k-1)}& \text{if } k-r+1\leq l\leq k-1\\
      \sum_{i=1}^nu_i\alpha_i^{n-1}-1& \text{if } l=k.
  \end{cases}.
 \end{align*}
  By Lemma \ref{paritylemma}, $g_lh_m^T=0$.
This completes the proof.
\end{proof}

\begin{theorem}\label{c1extend}
The code $C_1$ is an extended code of $C_{r,k}$, that is, $C_1=\overline{C}_{r,k}(w)$ for some non-zero $w\in\mathbb{F}_q^n$.
\end{theorem}
\begin{proof}
We prove that $G_1=(G_{r,k}\ G_{r,k}w^T)$ for some non-zero $w\in\mathbb{F}_q^n$. By Cramer's rule, we know that 
$$\begin{bmatrix}
    1&1& \dots &1&1\\
    \alpha_1& \alpha_2& \dots&\alpha_k&\alpha_{k+1}\\
    \alpha_1^2& \alpha_2^2& \dots&\alpha_k^2&\alpha_{k+1}^2\\
    \vdots &\vdots &\dots&\vdots&\vdots\\
    
    \alpha_1^{k-1}& \alpha_2^{k-1}& \dots&\alpha_k^{k-1}&\alpha_{k+1}^{k-1}\\
    \alpha_1^{k}& \alpha_2^{k}& \dots&\alpha_k^{k}&\alpha_{k+1}^k\\
\end{bmatrix}\begin{bmatrix}
    x_1\\
    x_2\\
    x_3\\
    \vdots\\
    x_{k}\\
      x_{k+1}\\
\end{bmatrix}=\begin{bmatrix}
    0\\
    0\\
    0\\
    \vdots\\
    0 \\
    1
\end{bmatrix}$$
has a non-zero solution $y=(y_1,\dots,y_k,y_{k+1})\in \mathbb{F}_q^{k+1}$, where 
$$y_i=\prod_{1\leq j\leq k+1,\ j\neq i}(\alpha_i-\alpha_j)^{-1}, \text{ for  } i=1,2,\dots,k+1.$$
Take $w=(w_1,w_2,\dots,w_n)\in \mathbb{F}_q^n$, where $w_i=y_i$ for $1\leq i\leq k+1$ and $w_i=0$ for $k+2\leq i\leq n$. Then $ w \neq \bm{0}$ and 
$$\begin{bmatrix}
    1&1& \dots &1&1&\dots &1\\
    \alpha_1& \alpha_2& \dots&\alpha_k&\alpha_{k+1}&\dots &\alpha_n\\
    \alpha_1^2& \alpha_2^2& \dots&\alpha_k^2&\alpha_{k+1}^2&\dots &\alpha_n^2\\
    \vdots &\vdots &\dots&\vdots&\vdots&\vdots\\
    
    \alpha_1^{k-1}& \alpha_2^{k-1}& \dots&\alpha_k^{k-1}&\alpha_{k+1}^{k-1}&\dots &\alpha_n^{k-1}\\
    \alpha_1^{k}& \alpha_2^{k}& \dots&\alpha_k^{k}&\alpha_{k+1}^k&\dots &\alpha_n^k\\
\end{bmatrix}\begin{bmatrix}
    w_1\\
    w_2\\
    w_3\\
    \vdots\\
    w_{k}\\
      w_{k+1}\\
      \vdots\\
      w_n
\end{bmatrix}=\begin{bmatrix}
    0\\
    0\\
    0\\
    \vdots\\
    0 \\
    1
\end{bmatrix}.$$

Consequently, 
$$\begin{bmatrix}
    1&1& \dots &1\\
    \alpha_1& \alpha_2& \dots&\alpha_n\\
    \alpha_1^2& \alpha_2^2& \dots&\alpha_n^2\\
    \vdots &\vdots &\dots&\vdots\\
    \alpha_1^{k-r-1}& \alpha_2^{k-r-1}& \dots&\alpha_n^{k-r-1}\\
    \alpha_1^{k-r+1}& \alpha_2^{k-r+1}& \dots&\alpha_n^{k-r+1}\\
    \vdots &\vdots &\dots&\vdots\\
    \alpha_1^{k-1}& \alpha_2^{k-1}& \dots&\alpha_n^{k-1}\\
    \alpha_1^{k}& \alpha_2^{k}& \dots&\alpha_n^{k}\\
\end{bmatrix}_{k\times n}\begin{bmatrix}
    w_1\\
    w_2\\
    w_3\\
    \vdots\\
    w_k\\
      w_{k+1}\\
      \vdots\\
      w_n
\end{bmatrix}=\begin{bmatrix}
    0\\
    0\\
    0\\
    \vdots\\
    0 \\
    1
\end{bmatrix},$$
that is, $G_{r,k}w^T=(0,0,\dots,1)_{k\times1}^T$. Hence, $G_1=(G_{r,k}\ G_{r,k}w^T).$ This completes the proof.
\end{proof}
The following result on MDSness of the code $C_1$ generated by $G_1$ is stated in \cite[Remark 2]{Li2024} without proof. For completeness, here, we provide a proof using Lemma \ref{mdsnonsingular}.
\begin{theorem}\label{c1mds}
Let $1\leq r\leq k-1$ and  $C_1$ be the code generated by $G_1$. Then $C_1$ is MDS if and only if 
    \begin{equation*}
    \begin{split}
         &\sum_{1\leq j_1<\dots <j_r\leq k}\alpha_{i_{j_1}}\alpha_{i_{j_2}}\dots\alpha_{i_{j_r}}\neq 0,\ \ \ \ \ \ \  \ \ \ \ \ \ \ \ \ (\ast)\\
          &\sum_{1\leq j_1<\dots <j_{r-1}\leq k-1}\alpha_{i_{j_1}}\alpha_{i_{j_2}}\dots\alpha_{i_{j_{r-1}}}\neq 0\ \ \ \ \ \ \ \   (\ast\ast)
        \end{split}
    \end{equation*}
    for all pairwise distinct elements $\alpha_{i_j}$ with $ \{i_1,i_2,\dots,i_{k-1}\}, \{i_1,i_2,\dots,i_{k}\}\subseteq \{1,2,\dots,n\}$. 
\end{theorem} 
\begin{proof}
Let $v=(0,0,\dots,1)^T$. Let $B_1$ be an arbitrary $k\times k $ sub-matrix of $G_1$, which does not contain $v$ as a column.  Then $B_1$ is a sub-matrix of $G_{r,k}$. Hence, $B_1$ is non-singular if and only if  $(\ast)$ holds. 
Let $r>1$ and $B_2$ be an arbitrary $k\times k$ sub-matrix of $G_1$ with $k-1$ columns from $G_{r,k}$ and containing $v$ as a column. Let 
$$B_2=\begin{bmatrix}
    1&1& \dots &1&0\\
    \alpha_{i_1}& \alpha_{i_2}& \dots&\alpha_{i_{k-1}}&0\\
    \alpha_{i_1}^2& \alpha_{i_2}^2& \dots&\alpha_{i_{k-1}}^2&0\\
    \vdots &\vdots &\dots&\vdots&\vdots\\
    \alpha_{i_1}^{k-r-1}& \alpha_{i_2}^{k-r-1}& \dots&\alpha_{i_{k-1}}^{k-r-1}&0\\
    \alpha_{i_1}^{k-r+1}& \alpha_{i_2}^{k-r+1}& \dots&\alpha_{i_{k-1}}^{k-r+1}&0\\
    \vdots &\vdots &\dots&\vdots&\vdots\\
    \alpha_{i_1}^{k-1}& \alpha_{i_2}^{k-1}& \dots&\alpha_{i_{k-1}}^{k-1}&0\\
    \alpha_{i_1}^{k}& \alpha_{i_2}^{k}& \dots&\alpha_{i_{k-1}}^{k}&1\\
\end{bmatrix}.$$
Then $$\det(B_2)=\det \begin{bmatrix}
    1&1& \dots &1\\
    \alpha_{i_1}& \alpha_{i_2}& \dots&\alpha_{i_{k-1}}\\
    \alpha_{i_1}^2& \alpha_{i_2}^2& \dots&\alpha_{i_{k-1}}^2\\
    \vdots &\vdots &\dots&\vdots\\
    \alpha_{i_1}^{k-r-1}& \alpha_{i_2}^{k-r-1}& \dots&\alpha_{i_{k-1}}^{k-r-1}\\
    \alpha_{i_1}^{k-r+1}& \alpha_{i_2}^{k-r+1}& \dots&\alpha_{i_{k-1}}^{k-r+1}\\
    \vdots &\vdots &\dots&\vdots\\
    \alpha_{i_1}^{k-1}& \alpha_{i_2}^{k-1}& \dots&\alpha_{i_{k-1}}^{k-1} 
    \end{bmatrix}.$$
    Thus, by Proposition \ref{detk-r},   $\det(B_2)\neq 0$ if and only $(\ast\ast)$ holds. If $r=1$, then the determinant of  matrix $B_2$ is the determinant of a Vandermonde matrix and $\det(B_2)\neq 0$.
\end{proof}
\begin{remark}\label{rem1}
For $r=1$, observe that Condition $(\ast\ast)$ holds trivially. Thus the code $C_1$ is MDS if and only if     \begin{equation*}
    \begin{split}
         &\sum_{1\leq j_1\leq k}\alpha_{i_{j_1}}\neq 0\ \ \ \ \ \ \  \ \ \ \ \ \ \ \ \ 
        \end{split}
    \end{equation*}
    for  $ \{i_1,i_2,\dots,i_{k}\}\subset \{1,2,\dots,n\}$. For $r=1$, in \cite{Li2024}, the authors explicitly studied codes of form $C_1$ and provided necessary and sufficient conditions for $C_1$ to be MDS, which is consistent with the above result. 
\end{remark}
\begin{theorem}
  Let $6\le 2k \le n$, $1\leq r\leq k-1$  and $C_1$ be an MDS code generated by $G_1$. Then, $C_1$ is a non-GRS MDS code.
\end{theorem}
\begin{proof}
    The proof follows from Lemma \ref{extendnonrs} and Theorem \ref{c1extend}.
\end{proof}

\begin{example}
 There exists a $[7,3]$ non-GRS MDS code over $\mathbb{F}_p$ for every $p>47$. For this,   take $r=2,$ $k=3$ and  $\alpha_{i}=i-1 \pmod{p}$ for $1\leq i\leq 6$. Then Conditions $(\ast)$ and $(\ast\ast)$ are satisfied for all primes $p>47$.   
\end{example}

In the following example, we construct a $[7,3]$ non-GRS MDS code over $\mathbb{F}_{17}$.
\begin{example}
    Let $q=17$, $n=6,$ $k=3$ and $r=2$. Take $\mathbf{\alpha}=(\alpha_1,\alpha_2,\alpha_3,\alpha_4,\alpha_5,\alpha_6)=(0,1,5,3,8,6)$. Then $\mathbf{\alpha}$ satisfies $(\ast)$ and $(\ast\ast)$ in Theorem \ref{c1mds}. Thus the 
    code generated $$G=\begin{bmatrix}
        1&1&1&1&1&1&0\\
\alpha_1^2&\alpha_2^2&\alpha_3^2&\alpha_4^2&\alpha_5^2&\alpha_6^2&0\\
         \alpha_1^3&\alpha_2^3&\alpha_3^3& \alpha_4^3&\alpha_5^3&\alpha_6^3&1\\
    \end{bmatrix}=\begin{bmatrix}
        1&1&1&1&1&1&0\\
        0&1&8&9&13&2&0\\
        0&1&6&10&2&12&1
    \end{bmatrix}$$ is a $[7,3,5]_{17}$ non-GRS MDS code. Verified by MAGMA.
\end{example}

In \cite{Jin2024}, some non-GRS MDS codes are constructed for $r=1$. Extending these codes we get the following corollaries. 
\begin{corollary}
    Let $p$ be an odd prime and $k,n$ be positive integers with $6\leq 2k\leq n\leq \frac{p}{k}+\frac{k+1}{2} $. Then there exists an $[n+1,k]$ non-GRS MDS code over $\mathbb{F}_p$.
\end{corollary}
\begin{proof}
    Take $r=1$ and $\alpha_i=i-1\pmod{p}$. Then Condition $(\ast)$ is satisfied. Thus by Remark \ref{rem1}, the code $C_1$ with parameters $[n+1,k]$ is a non-GRS MDS code over $\mathbb{F}_p$.
\end{proof}
\begin{corollary}
    Let $n$ and $k$ be positive integers with $3\leq k\leq n/2$. Let $C$ be an $[n,n-t,d\geq3]$ code over $\mathbb{F}_q$. If there is no codeword in $C$ with the Hamming weight $k$, then there exists an $[n+1,k]$ non-GRS MDS code over $\mathbb{F}_{q^t}$.
\end{corollary}

\begin{corollary}
    If there exists an $[n,n-t,d\geq k+1]$ code over $\mathbb{F}_q$, then there exists an $[n+1,k]$ non-GRS MDS code over $\mathbb{F}_{q^t}$.
\end{corollary}
\begin{example}
 It is known that there exists a $[11,4,6]$ ($n=11, t=7,k=5$) code over $\mathbb{F}_3$. Thus, by the above corollary,  there exists a $[12,5]$ non-GRS MDS code over $\mathbb{F}_{3^7}$. 
\end{example}

The following theorem provides necessary and sufficient conditions when an extended code of an MDS code remains MDS.
\begin{theorem}\cite[Theorem 1]{Wu2025}\label{coveradius}
  Let $C$ be an $[n,k]$ MDS code over $\mathbb{F}_q$. Then for any $w\in\mathbb{F}_q^n$, $\overline{C}(w)$  is MDS if and only if the covering radius of $C^{\perp}$ is $k$ and $w$ is a deep hole of the dual
 code  $C^{\perp}$.
\end{theorem}

Thus, by Theorem \ref{coveradius}, we have the following corollary.

\begin{corollary}
    Let $C_1$ be an $[n+1,k]$ MDS code generated by $G_1$. If $C_{r,k}$ is MDS, then the covering radius of $C_{r,k}^{\perp}$ is $k$ with deep whole vector $w$ obtained in Theorem \ref{c1extend}.
\end{corollary}

\section{Second class of extended non-GRS MDS codes}\label{secondextension}
In this section, we consider another class of codes whose generator matrix is obtained from the matrix $G_{r,k}$ by adding two columns. Analogous to the previous section, we show that these codes are extended codes of the codes considered in the previous section. We determine the parity check matrix for these codes and establish necessary and sufficient conditions for these codes to be non-GRS MDS codes. 
Let $\delta\in \mathbb{F}_q$ and
$$G_2=\begin{bmatrix}
    1&1& \dots &1 &0&0\\
    \alpha_1& \alpha_2& \dots&\alpha_n&0&0\\
    \alpha_1^2& \alpha_2^2& \dots&\alpha_n^2&0&0\\
    \vdots &\vdots &\dots&\vdots&\vdots&\vdots\\
    \alpha_1^{k-r-1}& \alpha_2^{k-r-1}& \dots&\alpha_n^{k-r-1}&0&0\\
    \alpha_1^{k-r+1}& \alpha_2^{k-r+1}&\dots&\alpha_n^{k-r+1}&0&0\\
    \vdots &\vdots &\dots&\vdots&\vdots& \vdots\\
    \alpha_1^{k-1}& \alpha_2^{k-1}& \dots&\alpha_n^{k-1}&0&1\\
    \alpha_1^{k}& \alpha_2^{k}& \dots&\alpha_n^{k}&1&\delta\\
\end{bmatrix}_{k\times (n+2)}$$ 
and $C_2$ be the code generated by $G_2$. The code $C_2$ is also seen as
\begin{equation*}
    \begin{split}
     C_2=&\{ (f(\alpha_1),f(\alpha_2),\dots, f(\alpha_n), f_k, f_{k-j}+\delta f_k)|\ f(x)\in V_{r,k} \},
    \end{split}
\end{equation*}
 where  $j=1$  if  $r>1$, $j=2$ if  $r=1$.
\begin{theorem}\label{parityg2}
Let $2\leq r\leq k-1$ and 
$$H_2=\begin{bmatrix}
    u_1&u_2&\dots&u_n&0&0\\
u_1\alpha_1&u_2\alpha_2&\dots&u_n\alpha_n&0&0\\
u_1\alpha_1^2&u_2\alpha_2^2&\dots&u_n\alpha_n^2&0\\
\vdots&\vdots&\vdots&\vdots&\vdots&\vdots\\
u_1\alpha_1^{n-k-2}&u_2\alpha_2^{n-k-2}&\dots&u_n\alpha_n^{n-k-2}&0&0\\
u_1\Lambda_{r,1}&u_2\Lambda_{r,2}&\dots&u_n\Lambda_{r,n}&0&0\\
u_1\alpha_1^{n-k-1}& u_2\alpha_2^{n-k-1}&\dots&u_n\alpha_n^{n-k-1}&-1&0\\
u_1\alpha_1^{n-k}& u_2\alpha_2^{n-k}&\dots&u_n\alpha_n^{n-k}& b &-1
\end{bmatrix},$$ where $b=\delta-\sum_{i=1}^n\alpha_i$.
Then $H_2$ is a parity check matrix for the code $C_2$ generated by $G_2$.
\end{theorem}
\begin{proof}
It is easy to see that the rows of $H_2$ are linearly independent, and the first $n-k+1$ rows of $H_2$ are orthogonal to each row of $G_2$ from Theorem \ref{parityg1}. Note that for $m=n-k+2$, we have 
\begin{align*}
    g_l h_m^T=\begin{cases}
    \sum_{i=1}^nu_i\alpha^{(l-1)+(n-k)}& \text{if } 1\leq l\leq k-r  \\
     \sum_{i=1}^nu_i\alpha^{l+(n-k)}& \text{if } k-r+1\leq l\leq k-2  \\
     \sum_{i=1}^nu_i\alpha^{n-1}-1& \text{if }  l= k-1\\
     \sum_{i=1}^nu_i\alpha^{n}+b-\delta& \text{if }  l= k.
    \end{cases}.
\end{align*}
By Lemma \ref{paritylemma}, $g_lh_m^T=0$.
This completes the proof.
\end{proof}

For the case $r=1$, in \cite{Zhi2025}, the authors determined the parity check matrix for the code $C_2$ generated by $G_2$. In the following proposition, we determine the parity check matrix for case $r=1$ analogous to the case $2\leq r\leq k-1$ in the above theorem. The proof follows similar to that of Theorem \ref{parityg2} using Lemma \ref{paritylemma}. Thus, we omit the proof.

\begin{proposition}\label{r1parity}
  Let $r=1$ and 
$$H_2=\begin{bmatrix}
    u_1&u_2&\dots&u_n&0&0\\
u_1\alpha_1&u_2\alpha_2&\dots&u_n\alpha_n&0&0\\
u_1\alpha_1^2&u_2\alpha_2^2&\dots&u_n\alpha_n^2&0\\
\vdots&\vdots&\vdots&\vdots&\vdots&\vdots\\
u_1\alpha_1^{n-k-2}&u_2\alpha_2^{n-k-2}&\dots&u_n\alpha_n^{n-k-2}&0&0\\
u_1\Lambda_{1,1}&u_2\Lambda_{1,2}&\dots&u_n\Lambda_{1,n}&0&0\\
u_1\alpha_1^{n-k-1}& u_2\alpha_2^{n-k-1}&\dots&u_n\alpha_n^{n-k-1}&-1&0\\
u_1\alpha_1^{n-k+1}& u_2\alpha_2^{n-k+1}&\dots&u_n\alpha_n^{n-k+1}& \delta-S_2 &-1
\end{bmatrix},$$
where $S_2:=S_2(\alpha_1,\alpha_2,\dots,\alpha_n)=(\sum_{i=1}^n\alpha_i)^2-\sum_{1\leq i<j\leq n}\alpha_i\alpha_j$. 
Then $H_2$ is a parity check matrix for the code $C_2$ generated by $G_2$.  
\end{proposition}

\begin{remark}
   Observe that the parity check matrix $H_{n-k+2}$ obtained in \cite{Zhi2025} and  the above parity check matrix $H_2$ in the Proposition \ref{r1parity}  generate the same code. Let $L_1,L_2$ and $L_3$ be the last, second last and third last rows of $H_2$, respectively. 
     Note that all rows are same in the both matrices except  the rows $L_3$ and second last row of $H_{n-k+2}$.      
    From Eq. \ref{lambdari}, we have  $\Lambda_{1,i}=\sigma_0\alpha_i^{n-k}-\sigma_1\alpha_i^{n-k-1}$. Apply a row operation on $L_3$, $$L_3\to L_3+\sigma_1L_2.$$  The updated row $L_3$ is
$$[u_1\alpha_1^{n-k}, u_2\alpha_2^{n-k},\dots, u_n\alpha_n^{n-k},-\sigma_1,0 ]$$
     which is same as the second last row of $H_{n-k+2}$. Thus, the matrix $H_2$ and $H_{n-k+2}$ generate the same code.
\end{remark}

\begin{theorem}\label{c2extend}
The code $C_2$ generated by $G_2$ is an extended code of $C_1$.
\end{theorem}
\begin{proof}
Similar to the proof of Theorem \ref{c1extend}, we have a non-zero $v\in \mathbb{F}_q^n$ such that $G_{r-1,k-1}v^T=(0,0,\dots,1)$, that is, 
$$\begin{bmatrix}
    1&1& \dots &1\\
    \alpha_1& \alpha_2& \dots&\alpha_n\\
    \alpha_1^2& \alpha_2^2& \dots&\alpha_n^2\\
    \vdots &\vdots &\dots&\vdots\\
    \alpha_1^{k-r-1}& \alpha_2^{k-r-1}& \dots&\alpha_n^{k-r-1}\\
    \alpha_1^{k-r+1}& \alpha_2^{k-r+1}& \dots&\alpha_n^{k-r+1}\\
    \vdots &\vdots &\dots&\vdots\\
    \alpha_1^{k-1}& \alpha_2^{k-1}& \dots&\alpha_n^{k-1}
\end{bmatrix}_{(k-1)\times n}\begin{bmatrix}
    v_1\\
    v_2\\
    v_3\\
    \vdots\\
    v_k\\
      v_{k+1}\\
      \vdots\\
      v_n
\end{bmatrix}=\begin{bmatrix}
    0\\
    0\\
    0\\
    \vdots\\
    0 \\
    1
\end{bmatrix}.$$ Take $w=(w_1,\dots,w_n,w_{n+1})\in \mathbb{F}_q^{n+1}$, where $w_i=v_i$ for $1\leq i\leq n$ and $w_{n+1}=\delta-\sum_{i=1}^n \alpha_i^k v_i$. Then $w\neq \bm{0}$ (as $v\neq \bm{0}$) and 
$$\begin{bmatrix}
    1&1& \dots &1&0\\
    \alpha_1& \alpha_2& \dots&\alpha_n&0\\
    \alpha_1^2& \alpha_2^2& \dots&\alpha_n^2&0\\
    \vdots &\vdots &\dots&\vdots&\vdots\\
    \alpha_1^{k-r-1}& \alpha_2^{k-r-1}& \dots&\alpha_n^{k-r-1}&0\\
    \alpha_1^{k-r+1}& \alpha_2^{k-r+1}& \dots&\alpha_n^{k-r+1}&0\\
    \vdots &\vdots &\dots&\vdots&\vdots\\
    \alpha_1^{k-1}& \alpha_2^{k-1}& \dots&\alpha_n^{k-1}&0\\
    \alpha_1^{k}& \alpha_2^{k}& \dots&\alpha_n^{k}&1\\
\end{bmatrix}_{k\times (n+1)}\begin{bmatrix}
    w_1\\
    w_2\\
    w_3\\
    \vdots\\
    w_k\\
      w_{k+1}\\
      \vdots\\
      w_n\\
      w_{n+1}
\end{bmatrix}=\begin{bmatrix}
    0\\
    0\\
    0\\
    \vdots\\
    0 \\
    1\\
    \delta
\end{bmatrix}.$$
Hence $(G_1\ G_1w^T)=G_2$. This completes the proof.
\end{proof}

\begin{theorem}\label{c2mds}
Let $C_2$ be the code generated by $G_2$ and $2\leq r\leq k-1$. Then    $C_2$ is MDS if and only if
\begin{equation*}
    \begin{split}
         &\sum_{1\leq j_1<\dots <j_r\leq k}\alpha_{i_{j_1}}\alpha_{i_{j_2}}\dots\alpha_{i_{j_r}}\neq 0,\ \ \ \ \ \ \  \ \ \ \ \ \ \ \ \ (\ast)\\
          &\sum_{1\leq j_1<\dots <j_{r-1}\leq k-1}\alpha_{i_{j_1}}\alpha_{i_{j_2}}\dots\alpha_{i_{j_{r-1}}}\neq 0,\ \ \ \ \ \ \ \   (\ast\ast)\\
          &\sum_{1\leq j_1<\dots <j_{r-2}\leq k-2}\alpha_{i_{j_1}}\alpha_{i_{j_2}}\dots\alpha_{i_{j_{r-2}}}\neq 0,\ \ \ \ \   \ \ \ (\ast\ast\ast)
        \end{split}
    \end{equation*}
    \begin{align*}
    \begin{split}
      &\delta \sum_{1\leq j_1<\dots <j_{r-1}\leq k-1}\alpha_{i_{j_1}}\alpha_{i_{j_2}}\dots\alpha_{i_{j_{r-1}}} - \left (\sum_{s=1}^{k-1}\alpha_{i_s}\right)\left(\sum_{1\leq j_1<\dots <j_{r-1}\leq k-1}\alpha_{i_{j_1}}\alpha_{i_{j_2}}\dots\alpha_{i_{j_{r-1}}}\right) \\
      &+\sum_{1\leq j_1<\dots <j_{r}\leq k-1}\alpha_{i_{j_1}}\alpha_{i_{j_2}}\dots\alpha_{i_{j_{r}}}\neq 0\ \ \ \ \ \ \ \ \ \ \ \ \ \ \ \ \ \  \ \ \   (\#) 
    \end{split}
\end{align*}
    for all pairwise distinct elements $\alpha_{i_j}$ with $\{i_1,i_2,\dots,i_{k-2}\}$,
    $\{i_1,i_2,\dots,i_{k-1}\}$,  
    $\{i_1,i_2,\dots,i_{k}\}\subseteq \{1,2,\dots,n\}$.
\end{theorem} 
\begin{proof}
 Let $u=(0,0,\dots,1)^T$ and $v=(0,0,\dots,1,\delta)^T$. Then, we have the following possible $k\times k$ sub-matrices of $G_2$.  \\
 1. Let $B_1$ be a $k\times k$ sub-matrix which does not contain $u$ and $v$. Then $B_1$ is a sub-matrix of $G_{r,k}$. Hence, $B_1$ is non-singular if and only if $(\ast)$ holds.\\
 2. Let $B_2$ be a $k\times k$ sub-matrix which  contain $u$ and $k-1$ columns of $G_{r,k}$. Then $B_2$ is a sub-matrix of $G_1$. Hence, $B_2$ is non-singular if and only if $(\ast\ast)$ holds.\\
 3.  Let $B_3$ be a $k\times k$ sub-matrix which  contain $v$ and $k-1$ columns of $G_{r,k}$. Assume that 
 $$B_3=\begin{bmatrix}
    1&1& \dots &1&0\\
    \alpha_{i_1}& \alpha_{i_2}& \dots&\alpha_{i_{k-1}}&0\\
    \alpha_{i_1}^2& \alpha_{i_2}^2& \dots&\alpha_{i_{k-1}}^2&0\\
    \vdots &\vdots &\dots&\vdots&\vdots\\
    \alpha_{i_1}^{k-r-1}& \alpha_{i_2}^{k-r-1}& \dots&\alpha_{i_{k-1}}^{k-r-1}&0\\
    \alpha_{i_1}^{k-r+1}& \alpha_{i_2}^{k-r+1}& \dots&\alpha_{i_{k-1}}^{k-r+1}&0\\
    \vdots &\vdots &\dots&\vdots&\vdots\\
    \alpha_{i_1}^{k-1}& \alpha_{i_2}^{k-1}& \dots&\alpha_{i_{k-1}}^{k-1}&1\\
    \alpha_{i_1}^{k}& \alpha_{i_2}^{k}& \dots&\alpha_{i_{k-1}}^{k}&\delta\\
\end{bmatrix}$$
for $\{i_1,i_2,\dots,i_{k-1}\}\subset \{1,2,\dots,n\}$. 

Let  $$B(y)=\begin{bmatrix}
    1&1& \dots &1&1\\
    \alpha_{i_1}& \alpha_{i_2}& \dots&\alpha_{i_{k-1}}&y\\
    \alpha_{i_1}^2& \alpha_{i_2}^2& \dots&\alpha_{i_{k-1}}^2&y^2\\
    \vdots &\vdots &\dots&\vdots&\vdots\\
    \alpha_{i_1}^{k-r-1}& \alpha_{i_2}^{k-r-1}& \dots&\alpha_{i_{k-1}}^{k-r-1}&y^{k-r-1}\\
    \alpha_{i_1}^{k-r+1}& \alpha_{i_2}^{k-r+1}& \dots&\alpha_{i_{k-1}}^{k-r+1}&y^{k-r+1}\\
    \vdots &\vdots &\dots&\vdots&\vdots\\
    \alpha_{i_1}^{k-1}& \alpha_{i_2}^{k-1}& \dots&\alpha_{i_{k-1}}^{k-1}&y^{k-1}\\
    \alpha_{i_1}^{k}& \alpha_{i_2}^{k}& \dots&\alpha_{i_{k-1}}^{k}&y^k\\
\end{bmatrix}.$$
Then $\det(B_3)=\delta A_{k}+A_{k-1}$, where $A_{k}$ and $A_{k-1}$ are the coefficients of $y^{k}$ and $y^{k-1}$, respectively in $\det(B(y))$. Observe that, by Proposition \ref{detk-r},
\begin{align*}
    \begin{split}
        \sum_{i=0}^kA_iy^i=& \prod_{1\leq s< t\leq k}(\alpha_{i_t}-\alpha_{i_s})\cdot \sum_{1\leq j_1<j_2<\dots<j_r\leq k}\alpha_{i_{j_1}}\alpha_{i_{j_2}}\dots \alpha_{i_{j_r}}\\ 
        =&\prod_{1\leq s< t\leq k-1}(\alpha_{i_t}-\alpha_{i_s}) \prod_{s=1}^{k-1}(y-\alpha_{i_s})\cdot \sum_{1\leq j_1<j_2<\dots<j_r\leq k}\alpha_{i_{j_1}}\alpha_{i_{j_2}}\dots \alpha_{i_{j_r}},
    \end{split}
\end{align*}
where $\alpha_{i_k}=y$.
Also,
\small \begin{align*}
    \begin{split}
        &\sum_{1\leq j_1<j_2<\dots<j_r\leq k}\alpha_{i_{j_1}}\alpha_{i_{j_2}}\dots \alpha_{i_{j_r}}\\
        &=\sum_{1\leq j_1<j_2<\dots<j_r\leq k-1}\alpha_{i_{j_1}}\alpha_{i_{j_2}}\dots \alpha_{i_{j_r}}+\sum_{1\leq j_1<j_2<\dots<j_{r-1}\leq k-1}\alpha_{i_{j_1}}\alpha_{i_{j_2}}\dots \alpha_{i_{j_{r-1}}\alpha_{i_k}}\\
        &=\sum_{1\leq j_1<j_2<\dots<j_r\leq k-1}\alpha_{i_{j_1}}\alpha_{i_{j_2}}\dots \alpha_{i_{j_r}}+\sum_{1\leq j_1<j_2<\dots<j_{r-1}\leq k-1}\alpha_{i_{j_1}}\alpha_{i_{j_2}}\dots \alpha_{i_{j_{r-1}}y}
    \end{split}
\end{align*}
Thus, the coefficient of $y^k$ is 
$$\prod_{1\leq s<t\leq k-1}(\alpha_{i_t}-\alpha_{i_s})\sum_{1\leq j_1<\dots <j_{r-1}\leq k-1,\ \ \ j_r=k}\alpha_{i_{j_1}}\alpha_{i_{j_2}}\dots\alpha_{i_{j_{r-1}}}$$
and coefficient of $y^{k-1}$ is 

    \begin{equation*}
    \begin{split}
        \prod_{1\leq s<t\leq k-1}(\alpha_{i_t}-\alpha_{i_s})\left [-\left (\sum_{s=1}^{k-1}\alpha_{i_s}\right)\left(\sum_{1\leq j_1<\dots <j_{r-1}\leq k-1}\alpha_{i_{j_1}}\alpha_{i_{j_2}}\dots\alpha_{i_{j_{r-1}}}\right) \right. \\ \left.+\sum_{1\leq j_1<\dots <j_{r}\leq k-1}\alpha_{i_{j_1}}\alpha_{i_{j_2}}\dots\alpha_{i_{j_{r}}}\right ].
         \end{split}
        \end{equation*} 
Hence, 
        \begin{equation*}
        \begin{split}
           \det (B_3)=&\left [\prod_{1\leq s<t\leq k-1}(\alpha_{i_t}-\alpha_{i_s})\right ]\cdot\\ &\left[ \delta \sum_{1\leq j_1<\dots <j_{r-1}\leq k-1}\alpha_{i_{j_1}}\alpha_{i_{j_2}}\dots\alpha_{i_{j_{r-1}}} - \left (\sum_{s=1}^{k-1}\alpha_{i_s}\right)\left(\sum_{1\leq j_1<\dots <j_{r-1}\leq k-1}\alpha_{i_{j_1}}\alpha_{i_{j_2}}\dots\alpha_{i_{j_{r-1}}}\right) \right.\\ 
      &\left. +\sum_{1\leq j_1<\dots <j_{r}\leq k-1}\alpha_{i_{j_1}}\alpha_{i_{j_2}}\dots\alpha_{i_{j_{r}}}  \right].
      \end{split}
     \end{equation*}

Thus, $\det(B_3)\neq0$ if and only if $(\#)$ holds. \\

4. Let $B_4$ be a $k\times k$ sub-matrix which  contains $u$, $v$ and $k-2$ columns of $G_{r,k}$. Assume that 
 $$B_4=\begin{bmatrix}
    1&1& \dots &1&0&0\\
    \alpha_{i_1}& \alpha_{i_2}& \dots&\alpha_{i_{k-2}}&0&0\\
    \alpha_{i_1}^2& \alpha_{i_2}^2& \dots&\alpha_{i_{k-2}}^2&0&0\\
    \vdots &\vdots &\dots&\vdots&\vdots&\vdots\\
    \alpha_{i_1}^{k-r-1}& \alpha_{i_2}^{k-r-1}& \dots&\alpha_{i_{k-2}}^{k-r-1}&0&0\\
    \alpha_{i_1}^{k-r+1}& \alpha_{i_2}^{k-r+1}& \dots&\alpha_{i_{k-2}}^{k-r+1}&0&0\\
    \vdots &\vdots &\dots&\vdots&\vdots\\
    \alpha_{i_1}^{k-1}& \alpha_{i_2}^{k-1}& \dots&\alpha_{i_{k-2}}^{k-1}&0&1\\
    \alpha_{i_1}^{k}& \alpha_{i_2}^{k}& \dots&\alpha_{i_{k-2}}^{k}&1&\delta\\
\end{bmatrix}$$
for $\{i_1,i_2,\dots,i_{k-2}\}\subset \{1,2,\dots,n\}$. 
Note that for $r>2$, we have
$$\det(B_4)=(-1)^{k-1+k}\det \begin{bmatrix}
    1&1& \dots &1\\
    \alpha_{i_1}& \alpha_{i_2}& \dots&\alpha_{i_{k-2}}\\
    \alpha_{i_1}^2& \alpha_{i_2}^2& \dots&\alpha_{i_{k-2}}^2\\
    \vdots &\vdots &\dots&\vdots\\
    \alpha_{i_1}^{k-r-1}& \alpha_{i_2}^{k-r-1}& \dots&\alpha_{i_{k-2}}^{k-r-1}\\
    \alpha_{i_1}^{k-r+1}& \alpha_{i_2}^{k-r+1}& \dots&\alpha_{i_{k-2}}^{k-r+1}\\
    \vdots &\vdots &\dots&\\
    \alpha_{i_1}^{k-2}& \alpha_{i_2}^{k-2}& \dots&\alpha_{i_{k-2}}^{k-2}   
\end{bmatrix}.$$
Thus, by Proposition \ref{detk-r}, $\det (B_4)$ is nonzero if and only if $(\ast\ast\ast)$ holds. If $r=2$, then the determinant of matrix $B_4$ is the determinant of a Vandermonde matrix and $\det(B_4)\neq0$. This completes the proof.
\end{proof}

For the case $r=1$, in \cite{Zhi2025}, the authors provided the necessary and sufficient conditions for $C_2$ to be MDS. In the following proposition, we demonstrate that the proof for the case $r=1$ can be done similar to the case $r\geq 2$ as in the above Theorem \ref{c2mds}.

\begin{proposition}\label{c2mdsreq1}
Let $C_2$ be the code generated by $G_2$ and $r=1$. Then    $C_2$ is MDS if and only if
\begin{equation*}
    \begin{split}
         \sum_{1\leq j_1\leq k}\alpha_{i_{j_1}}\alpha_{i_{j_2}}\dots\alpha_{i_{j_r}}\neq 0
        \end{split}
    \end{equation*}
and  $$\delta  - \left (\sum_{s=1}^{k-1}\alpha_{i_s}\right)^2
      +\sum_{1\leq j_1<j_{2}\leq k-1}\alpha_{i_{j_1}}\alpha_{i_{j_2}}\neq 0.$$
    for all pairwise distinct elements $\alpha_{i_j}$ with $\{i_1,i_2,\dots,i_{k-1}\}$, $\{i_1, i_2,\dots,i_{k}\}\subseteq \{1,2,\dots,n\}$.
 \end{proposition}
 \begin{proof}
Let $B_1,\ B_2,\ B_3$ and $B_4$ be the matrices as in the proof of the Theorem \ref{c2mds} corresponding to the possible cases. Then it is easy to see that $B_2$ and $B_4$ are non-singular, and $B_1$ is non-singular if and only if $$\sum_{1\leq j_1\leq k}\alpha_{i_{j_1}}\alpha_{i_{j_2}}\dots\alpha_{i_{j_r}}\neq 0.$$ 
 Moreover, the matrix $B_3$ will be of the form,
  $$B_3=\begin{bmatrix}
    1&1& \dots &1&0\\
    \alpha_{i_1}& \alpha_{i_2}& \dots&\alpha_{i_{k-1}}&0\\
    \alpha_{i_1}^2& \alpha_{i_2}^2& \dots&\alpha_{i_{k-1}}^2&0\\
    \vdots &\vdots &\dots&\vdots&\vdots\\
    \alpha_{i_1}^{k-2}& \alpha_{i_2}^{k-2}& \dots&\alpha_{i_{k-1}}^{k-2}&1\\
    \alpha_{i_1}^{k}& \alpha_{i_2}^{k}& \dots&\alpha_{i_{k-1}}^{k}&\delta\\
\end{bmatrix}$$
for some $\{i_1,i_2,\dots,i_{k-1}\}\subset \{1,2,\dots,n\}$. 

Let  $$B(y)=\begin{bmatrix}
    1&1& \dots &1&1\\
    \alpha_{i_1}& \alpha_{i_2}& \dots&\alpha_{i_{k-1}}&y\\
    \alpha_{i_1}^2& \alpha_{i_2}^2& \dots&\alpha_{i_{k-1}}^2&y^2\\
    \vdots &\vdots &\dots&\vdots&\vdots\\
    \alpha_{i_1}^{k-2}& \alpha_{i_2}^{k-2}& \dots&\alpha_{i_{k-1}}^{k-2}&y^{k-2}\\
    \alpha_{i_1}^{k}& \alpha_{i_2}^{k}& \dots&\alpha_{i_{k-1}}^{k}&y^k\\
\end{bmatrix}.$$
Then $\det(B_3)=\delta A_{k}+A_{k-2}$, where $A_{k}$ and $A_{k-2}$ are the coefficients of $y^{k}$ and $y^{k-2}$, respectively in $\det(B(y))$. Observe that, by Proposition \ref{detk-r},
\begin{align*}
    \begin{split}
        \sum_{i=0}^kA_iy^i=& \prod_{1\leq s< t\leq k}(\alpha_{i_t}-\alpha_{i_s})\cdot (\alpha_{i_1}+\alpha_{i_2}+\dots+\alpha_{i_{k-1}}+y)\\ 
        =&\prod_{1\leq s< t\leq k-1}(\alpha_{i_t}-\alpha_{i_s}) \prod_{s=1}^{k-1}(y-\alpha_{i_s})\cdot (\alpha_{i_1}+\alpha_{i_2}+\dots+\alpha_{i_{k-1}}+y),
    \end{split}
\end{align*}
where $\alpha_{i_k}=y$.

Thus, the coefficient of $y^k$ is 
$$\prod_{1\leq s<t\leq k-1}(\alpha_{i_t}-\alpha_{i_s})$$
and the coefficient of $y^{k-2}$ is 
\small \begin{align*}
   \begin{split}
       \prod_{1\leq s<t\leq k-1}(\alpha_{i_t}-\alpha_{i_s})\left [-\left (\sum_{s=1}^{k-1}\alpha_{i_s}\right)(\alpha_{i_1}+\alpha_{i_2}+\dots+\alpha_{i_{k-1}}) +\sum_{1\leq j_1< j_{2}\leq k-1}\alpha_{i_{j_1}}\alpha_{i_{j_2}}\right ].
   \end{split} 
\end{align*}

Hence, \small \begin{align*}
    \begin{split}
      \det (B_3)=  \left [\prod_{1\leq s<t\leq k-1}(\alpha_{i_t}-\alpha_{i_s})\right ]
      \left (\delta  - \left (\sum_{s=1}^{k-1}\alpha_{i_s}\right)^2
      +\sum_{1\leq j_1<j_{2}\leq k-1}\alpha_{i_{j_1}}\alpha_{i_{j_2}}\right). 
    \end{split}
\end{align*}
Thus, $\det(B_3)\neq 0$ if and only 
$$\delta  - \left (\sum_{s=1}^{k-1}\alpha_{i_s}\right)^2
      +\sum_{1\leq j_1<j_{2}\leq k-1}\alpha_{i_{j_1}}\alpha_{i_{j_2}}\neq 0.$$   
      This completes the proof.     
\end{proof}

\begin{theorem}
Let $6\leq 2k\leq n$, $1\leq r\leq k-1$ and  $C_2$ be an MDS code generated by $G_2$. Then  $C_2$ is a non-GRS MDS code.
\end{theorem} 
\begin{proof}
    The desired conclusion follows from Lemma \ref{extendnonrs} and Theorem \ref{c2extend}.
\end{proof}

\begin{example}
    Let $q=17$, $n=6,$ $k=3$ and $r=2$. Take $\mathbf{\alpha}=(\alpha_1,\alpha_2,\alpha_3,\alpha_4,\alpha_5,\alpha_6)=(0,1,5,3,8,6)$. Then $\mathbf{\alpha}$ satisfies $(\ast)$ and $(\ast\ast)$ in Theorem \ref{c1mds}. 
    By Condition $(\#)$, for $\delta\in \{0,2,9,12,14\}$, the
    code generated by $$G=\begin{bmatrix}
        1&1&1&1&1&1&0&0\\
\alpha_1^2&\alpha_2^2&\alpha_3^2&\alpha_4^2&\alpha_5^2&\alpha_6^2&0&1\\
         \alpha_1^3&\alpha_2^3&\alpha_3^3& \alpha_4^3&\alpha_5^3&\alpha_6^3&1&\delta\\
    \end{bmatrix}=\begin{bmatrix}
        1&1&1&1&1&1&0&0\\
        0&1&8&9&13&2&0&1\\
        0&1&6&10&2&12&1&\delta
    \end{bmatrix}$$ is a $[8,3,6]_{17}$ non-GRS MDS code. Verified by MAGMA.
\end{example}

Following \cite{Wu2025}, by Theorem \ref{coveradius}, we have the following corollary. 
\begin{corollary}
    Let $C_2$ be an $[n+2,k]$ MDS code generated by the matrix $G_2$. If $C_1$ is an MDS code generated by the matrix $G_1$, then the covering radius of $C_{1}^{\perp}$ is $k$ with deep whole vector $w$ obtained in Theorem \ref{c2extend}.
\end{corollary}

\section{MDS codes and $o$-monomials}\label{oplynomials}
In this section, we discuss a connection between MDS codes of length $q+2$ over $\mathbb{F}_q$  and hyperovals in $\text{PG}(2,q)$, where $q=2^m$ for some positive integer $m$. We provide a new characterization of $o$-monomials in terms of the complete symmetric functions. Throughout this section, we assume $q=2^m$ for some positive integer $m$. For more details on $o$-polynomials and hyperovals, the reader is referred to \cite{Abdukhalikov2017,Abdukhalikov2019,Ball2019,Hirschfeld1998,Kiss2023}.
\begin{definition}
    In a plane of order $q$ a $(q +1)$-arc is called an oval, and a $(q+2)$-arc is called a hyperoval.
\end{definition}
\begin{definition}
    A polynomial $f(x)$ in $\mathbb{F}_q[x]$ is called an $o$-polynomial if 
    $$D(f)=\{(1:x:f(x))\mid  x\in \mathbb{F}_q\}\cup \{(0:1:0),(0:0:1)\}$$
    is a hyperoval.
\end{definition}
Suppose  $\mathbb{F}_q=\{\alpha_1,\alpha_2,\alpha_3,\dots,\alpha_q\}$.
The following theorem is well known \cite{Hirschfeld1998}.
\begin{theorem}
Let $C_o$ be a code of length $q+2$ over $\mathbb{F}_q$ with a generator matrix 
\begin{equation}\label{opolymatrix}
   G_o= \begin{bmatrix}
    1&1&1&\dots &1&0&0\\
    \alpha_1&\alpha_2&\alpha_3&\dots&\alpha_q&1&0\\
f(\alpha_1)&f(\alpha_2)&f(\alpha_3)&\dots&f(\alpha_q)&0&1
\end{bmatrix}
\end{equation}
for some $f(x)\in \mathbb{F}_q[x]$. Then $C$ is an MDS code if and only if $f(x)$ is an $o$-polynomial.
\end{theorem}
The columns of the matrix $G_o$ represent the points in projective space $\text{PG}(2,q)$. Thus, the set of columns of $G_o$ represents a hyperoval in $\text{PG}(2,q)$ if and only if $C_o$ is MDS.    
For example, let $f(x)=x^2$. Then $C_o$ is a double extended RS code. Therefore, $C_o$ is MDS. Hence $f(x)=x^2$ is an $o$-polynomial. 

Only a limited number of $o$-polynomial families have been identified so far(see \cite[Table 1]{Ball2019}). A complete classification of $o$-polynomials remains an open problem, and the discovery of new $o$-polynomial classes continues to be a difficult and open area of research. It is known \cite{Caullery2015} that if $f(x)$ is an $o$-polinomial of degree less than $\frac{1}{2}q^{1/4}$ then $f(x)$ is equivalent to either $x^6$ or $x^h$ for a positive integer $h$. In the next theorem, we provide a characterization of $o$-monomials in terms of the complete symmetric functions.

\begin{theorem}
    Let $f(x)=x^h$ for some positive integer $h$ and  $\gcd(h,q-1)=1$. Then $f(x)$ is an $o$-polynomial if and only if $S_{h-2}(\alpha_{i_1},\alpha_{i_2},\alpha_{i_3})\neq 0$ for all pairwise distinct $\alpha_{i_j}$ with $\{\alpha_{i_1},\alpha_{i_2},\alpha_{i_3}\}\subseteq\mathbb{F}_q$.
\end{theorem}
\begin{proof}
 Let $f(x)=x^h$ be an $o$-polynomial. Then by Theorem \ref{opolymatrix}, the code $C_o$ generated by $G_o$ is an MDS code. Consequently, all $3\times 3$ submatrices of $G_o$ are non-singular. Let $B$ be a $3\times 3$ submatrix of $G_o$ not containing the last two columns of $G_o$, that is,
 $$B=\begin{bmatrix}
1&1&1\\
\alpha_{i_1}& \alpha_{i_2}&\alpha_{i_3}\\
\alpha_{i_1}^h&\alpha_{i_2}^h&\alpha_{i_3}^h
 \end{bmatrix}$$
 for some $\{\alpha_{i_1},\alpha_{i_2},\alpha_{i_3}\}\subseteq\mathbb{F}_q$. By Lemma \ref{detgh}, we have $$\det(B)=S_{h-2}(\alpha_{i_1},\alpha_{i_2},\alpha_{i_3})\cdot\prod_{1\leq j_1<j_2\leq 3}(\alpha_{i_{j_2}}-\alpha_{i_{j_1}}).$$
 Since $\det(B)\neq 0$ and $\{\alpha_{i_1},\alpha_{i_2},\alpha_{i_3}\}\subseteq\mathbb{F}_q$ is arbitrary, therefore $S_{h-2}(\alpha_{i_1},\alpha_{i_2},\alpha_{i_3})\neq 0$ for all pairwise distinct $\{\alpha_{i_1},\alpha_{i_2},\alpha_{i_3}\}\subseteq\mathbb{F}_q$.

 Conversely, let $S_{h-2}(\alpha_{i_1},\alpha_{i_2},\alpha_{i_3})\neq 0$ for all pairwise distinct $\{\alpha_{i_1},\alpha_{i_2},\alpha_{i_3}\}\subseteq\mathbb{F}_q$. We show that the code $C_o$ generated by $G_o$ is MDS. Let $B$ be an arbitrary $3\times 3$ submatrix of $G_o$. We have the following possibilities:\\
 1. Let  $$B=\begin{bmatrix}
1&0&0\\
\alpha_{i_1}& 1&0\\
\alpha_{i_1}^h&0&1
 \end{bmatrix}$$ for some $\alpha_{i_1}\in \mathbb{F}_q$. It is easy to see that $\det(B)\neq 0$.\\
 2. Let  $$B=\begin{bmatrix}
1&1&0\\
\alpha_{i_1}& \alpha_{i_2}&1\\
\alpha_{i_1}^h&\alpha_{i_2}^h&0
 \end{bmatrix}$$ for some $\{\alpha_{i_1},\alpha_{i_2}\}\subseteq\mathbb{F}_q$. Then $\det(B)=\alpha_{i_1}^h-\alpha_{i_2}^h\neq0$.\\
 3. Let $$B=\begin{bmatrix}
1&1&0\\
\alpha_{i_1}& \alpha_{i_2}&0\\
\alpha_{i_1}^h&\alpha_{i_2}^h&1
 \end{bmatrix}$$ for some $\{\alpha_{i_1},\alpha_{i_2}\}\subseteq\mathbb{F}_q$. Then $\det(B)=\alpha_{i_2}-\alpha_{i_1}\neq0$.\\
 4.  $$B=\begin{bmatrix}
1&1&1\\
\alpha_{i_1}& \alpha_{i_2}&\alpha_{i_3}\\
\alpha_{i_1}^h&\alpha_{i_2}^h&\alpha_{i_3}^h
 \end{bmatrix}$$ for some $\{\alpha_{i_1},\alpha_{i_2},\alpha_{i_3}\}\subseteq\mathbb{F}_q$. Then, by Lemma \ref{detgh} $$\det(B)=S_{h-2}(\alpha_{i_1},\alpha_{i_2},\alpha_{i_3})\cdot\prod_{1\leq j_1<j_2\leq 3}(\alpha_{i_{j_2}}-\alpha_{i_{j_1}})\neq 0.$$\\
 Thus, by Lemma \ref{mdsnonsingular}, the code $C_o$ is MDS. Hence $f(x)=x^h$ is an $o$-polynomial.
\end{proof}

\section*{Conclusion}
Maximum distance separable (MDS) codes are highly valued for their practical applications, with most known constructions being generalized Reed-Solomon (GRS) codes. MDS codes that are not equivalent to GRS codes are called non-GRS MDS codes. Constructing non-GRS MDS codes remains a challenging and intriguing problem in coding theory, often limited by stringent parameter constraints. In this work, we constructed two classes of non-GRS MDS codes. Furthermore, we determined the parity check matrices for these codes. Also, we used the connection of MDS codes with arcs in finite projective spaces and provided a new characterization of $o$-monomials in terms of the complete symmetric functions.

\section*{Acknowledgment}

The authors were supported by the UAEU-AUA grant G00004614.

\bibliographystyle{abbrv}
	\bibliography{ref}

\end{document}